\documentclass[conference,a4paper]{IEEEtran}
\onecolumn
\usepackage[utf8]{inputenc} 
\usepackage[T1]{fontenc}
\usepackage{url}              
\usepackage{cite}             

\usepackage[cmex10]{amsmath}  
\usepackage{mleftright}       
\mleftright                   

\usepackage{graphicx}         
\usepackage{booktabs}         

\usepackage[utf8]{inputenc}
\usepackage[utf8]{inputenc} 
\usepackage[T1]{fontenc}
\usepackage{url}
\usepackage{ifthen}
\usepackage{cite}
\usepackage{bbm, dsfont}
\usepackage[cmex10]{amsmath} 

\usepackage{cite}
\usepackage{amsmath,amssymb,amsfonts}
\usepackage{mathtools}
\usepackage{amsthm}
\usepackage{algorithmic}
\usepackage{graphicx}
\usepackage{textcomp}
\usepackage{tikz}
\usepackage{caption}
\usepackage{cuted}
\usepackage{romannum}
\usepackage[utf8]{inputenc}
\usepackage{pgfplots} 
\usepackage{pgfgantt}
\usepackage{pdflscape}
\usepackage{amssymb}
\usepackage{comment}
\usepackage{pst-plot}
\usetikzlibrary{spy}
\usetikzlibrary{positioning,calc}
\usetikzlibrary{decorations.pathmorphing,calc,shapes,shapes.geometric,patterns}
\usetikzlibrary{shapes.multipart}
\usepackage{xfrac}
\usepackage{colortbl}
\usepackage{cancel}
\usetikzlibrary{arrows,positioning,calc,intersections}
\usetikzlibrary{datavisualization.formats.functions}
\def\BibTeX{{\rm B\kern-.05em{\sc i\kern-.025em b}\kern-.08em
    T\kern-.1667em\lower.7ex\hbox{E}\kern-.125emX}}
    
\usepackage{pgfplots}
\usepgfplotslibrary{fillbetween}
\usetikzlibrary{arrows, decorations.markings}
\newtheorem{theorem}{Theorem}
\newtheorem{lemma}{Lemma}

\newtheorem{definition}{Definition}
\newtheorem{claim}{Claim}
\newtheorem{claimproof}{Proof of Claim}

\newcommand{\setd}{\ensuremath{\mathcal{D}}}

\newcommand{\bs}[1]{\boldsymbol{#1}}

\definecolor{calpolypomonagreen}{rgb}{0.12, 0.3, 0.17}
\newcounter{remarkcount}

\newcommand{\circlearrow}{}
\DeclareRobustCommand{\circlearrow}{%
  \mathrel{\vphantom{\rightarrow}\mathpalette\circle@arrow\relax}%
}
\newcommand{\circle@arrow}[2]{%
  \m@th
  \ooalign{%
    \hidewidth$#1\circ\mkern1mu$\hidewidth\cr
    $#1-$\cr}%
}
\makeatother
\let\emptyset\varnothing
\usepackage{amsmath, amssymb, amsfonts, amsthm}
\usepackage{bm}
\usepackage{xcolor}
\usepackage{framed}

\usepackage{pgf}
\usepackage{tikz}
\usetikzlibrary{arrows,automata}
\usetikzlibrary{positioning}
\usetikzlibrary{decorations.pathmorphing}
\usetikzlibrary{shapes.geometric}
\usetikzlibrary{fit}
\usetikzlibrary{backgrounds}

\usepackage[caption=false,font=footnotesize]{subfig}

\newcommand{\mbf}{\mathbf}
\newcommand{\mc}{\mathcal}
\newcommand{\mbb}{\mathbb}

\theoremstyle{definition}

\theoremstyle{remark}

\def\BibTeX{{\rm B\kern-.05em{\sc i\kern-.025em b}\kern-.08em
    T\kern-.1667em\lower.7ex\hbox{E}\kern-.125emX}}
\begin{document}

\title{A Lower and Upper Bound on the Epsilon-Uniform Common Randomness Capacity } 

\color{black}\author{
\IEEEauthorblockN{Rami Ezzine\IEEEauthorrefmark{1}\IEEEauthorrefmark{3}, Moritz Wiese\IEEEauthorrefmark{1}\IEEEauthorrefmark{3}, Christian Deppe\IEEEauthorrefmark{1}\IEEEauthorrefmark{3}\IEEEauthorrefmark{5} and Holger Boche\IEEEauthorrefmark{1}\IEEEauthorrefmark{2}\IEEEauthorrefmark{3}\IEEEauthorrefmark{4}\IEEEauthorrefmark{5}}
\IEEEauthorblockA{\IEEEauthorrefmark{1}Technical University of Munich, Munich, Germany\\
\IEEEauthorrefmark{2}CASA -- Cyber Security in the Age of Large-Scale Adversaries–
Exzellenzcluster, Ruhr-Universit\"at Bochum, Germany\\
\IEEEauthorrefmark{3}BMBF Research Hub 6G-life, Munich, Germany\\
\IEEEauthorrefmark{4}{\color{black}{ Munich Center for Quantum Science and Technology (MCQST) }}\\
\IEEEauthorrefmark{5} {\color{black}{Munich Quantum Valley (MQV)}} \\
Email: \{rami.ezzine, wiese, christian.deppe, boche\}@tum.de}
}\color{black}

\maketitle
\thispagestyle{plain}
\pagenumbering{arabic}
\pagestyle{plain}
\begin{abstract}
 We consider a standard two-source model for uniform common randomness (UCR) generation, in which Alice and Bob observe independent and identically distributed (i.i.d.) samples of a  correlated finite source and where Alice is allowed to send information to Bob over an arbitrary single-user channel. We study the \(\boldsymbol{\epsilon}\)-UCR capacity for the proposed model, defined as the maximum common randomness rate one can achieve such that the probability that Alice and Bob do not agree on a common uniform or nearly uniform random variable does not exceed \(\boldsymbol{\epsilon}.\) 
 We establish a lower and \color{black} an \color{black} upper bound on the \(\boldsymbol{\epsilon}\)-UCR capacity using the bounds on the \(\boldsymbol{\epsilon}\)-transmission capacity proved by Verd\'u and Han for arbitrary point-to-point channels.
\end{abstract}

\section{Introduction}
\label{introduction}

Common Randomness (CR) is a highly valuable resource for modern communication systems. It is expected that the robustness, low-latency, ultra-reliability, resilience and security requirements imposed by these communication systems will be met on the basis of CR \cite{6Gcomm}\cite{6Gpostshannon}. In the CR generation framework, the sender Alice and the receiver Bob, often described as terminals, aim to agree on a common random variable with high probability \cite{part2}.

\color{black}CR allows an enormous performance gain in Post-Shannon communication tasks such as identification and secure identification \cite{part2}\cite{Generaltheory}\cite{CRincrease}, which are key techniques for the 6G technology\cite{keytechniques}.  
The identification scheme \cite{identification} is a new approach in communications developed by Ahlswede and Dueck in 1989. \color{black}
It is more efficient than the classical transmission scheme proposed by Shannon \cite{shannon} in several applications such as machine-to-machine
and human-to-machine systems \cite{applications}, industry 4.0 \cite{industrie40} and 6G communication systems \cite{6Gcomm}, which require ultra-reliable low-latency information exchange. Further applications of the identification scheme include digital watermarking \cite{Moulin,watermarkingahlswede, watermarking}.

The resource CR is also of high relevance in cryptography since under additional secrecy constraints, the generated CR can be used as secret keys, as shown in the fundamental two papers \cite{part1}\cite{Maurer}.
 CR is also highly relevant in the modular coding scheme for secure communication, where the generated randomness can be
used as a seed \cite{semanticsecurity}.

 CR is a useful resource for coding over arbitrarily varying channels \cite{capacityAVC}\cite{bandfive}, where we require only a little amount of CR compared to the set of messages. 
 By adding CR, one can fully compensate the active jamming attacks. Therefore, resilience by design can be achieved \cite{6Gcomm}. It is in this context worth mentioning that the security and resilience requirements are crucial for achieving trustworthiness. The latter represents a major challenge for future communication systems \cite{6Gandtrustworthiness}.

Different information theoretical models for CR generation have been investigated in the literature \cite{part2}\cite{survey}. The most standard one is a two-source model with unidirectional communication introduced by Ahlswede and Csiszár in \cite{part2}. In the two-source model, Alice and Bob observe independent and identically distributed (i.i.d.) samples of a correlated finite source. In \cite{part2}, the authors considered first the case when the two terminals are allowed to communicate over perfect channels and second the case when the terminals communicate over discrete noisy channels. They derived a single-letter formula of the CR capacity for both scenarios. It was additionally shown that for the proposed models, the CR capacity and the uniform common randomness (UCR) capacity are asymptotically the same and that the CR capacity can be always attained with nearly uniform random variables.
This is, from a practical perspective, the most convenient form of CR.
A more general scenario has been investigated in \cite{generalformulaUCRcapacity}, where Alice is allowed to send information to Bob via an arbitrary single-user channel. The authors in \cite{generalformulaUCRcapacity} established a general formula for the UCR capacity by making use of  a general formula for the channel transmission capacity elaborated in \cite{verduhan}.

\color{black}
We consider the two-source model for CR generation with one-way communication over arbitrary single-user channels proposed in \cite{generalformulaUCRcapacity}, where no further assumption on
stationarity, ergodicity or any kind of information stability is imposed.  We study  the $\epsilon$-UCR capacity for the proposed model, defined as the maximum rate of UCR one can attain such that the probability that Alice and Bob do not agree on a common uniform or nearly uniform random variable does not exceed $\epsilon,$ where unlike in \cite{generalformulaUCRcapacity},  $0<\epsilon<1$ is now fixed and cannot be made arbitrarily small.
\color{black}

The main contribution of our work consists in establishing general bounds on the $\epsilon$-UCR capacity that hold for arbitrary point-to-point channels. 
In our proof of the bounds on the $\epsilon$-UCR capacity, we make use of a well-known result of \cite{verduhan}, which is a lower and upper on the $\epsilon$-transmission capacity of arbitrary channels based on the inf-information rate between the channel inputs and outputs. The bounds that we derive on the $\epsilon$-UCR capacity hold with equality except possibly at the points of discontinuity of the $\epsilon$-transmission capacity, of which there are, at most, countably many. 

 For the sake of notational simplicity, we assume throughout the paper that the channel input and output alphabets are finite. 
 
\textit{Outline:} The remainder of the paper is structured as follows. In Section \ref{sec2}, we present our system model for CR generation, review the definition of an achievable $\epsilon$-transmission rate and the $\epsilon$-transmission capacity as well as the definition of an achievable $\epsilon$-UCR rate and the $\epsilon$-UCR capacity and present our main result. In Section \ref{directpart}, we prove the lower bound on the $\epsilon$-UCR capacity. Section \ref{converseproof} is dedicated to the proof of the upper-bound on the $\epsilon$-UCR capacity, where we use a change of measure argument introduced in \cite{strongconverse}. Section \ref{conclusion} contains concluding remarks and proposes a potential future work in this field. The proofs of several auxiliary lemmas are collected in the appendix.

\textit{Notation:} Throughout the paper,  $\log$ is taken to  base 2 and $\ln$ refers to the natural logarithm. For any set $\mc E,$ $\mc E^c$ refers to its complement and $\lvert \mc E \rvert$ refers to its cardinality.  For any random variable $X$ with distribution $P_X,$ $\text{supp}(P_X)$ refers to its support. For any random variables $X$ and $Y$ with respective distribution $P_X$ and $P_Y,$ $D(P_X;P_Y)$ denotes the relative entropy from $P_Y$ to $P_X.$   
\section{System Model, Definitions and Main Result}
\label{sec2}
\subsection{System Model}
\label{systemmodel}
Let a discrete memoryless multiple source (DMMS) $P_{XY}$ with two components, with  generic variables $X$ and $Y$ on alphabets $\mathcal{X}$ and $\mathcal{Y}$, respectively, be given. The DMMS emits i.i.d. samples of $(X,Y).$
Suppose that the outputs of $X$ are observed only by Terminal $A$ and those of $Y$ only by Terminal $B.$ Assume also that the joint distribution of $(X,Y)$ is known to both terminals.
Terminal $A$
can communicate with Terminal $B$ over an arbitrary single-user channel $\mbf W=\{W_n: \mc T^n \rightarrow \mc Z^n\}_{n=1}^{\infty},$ defined as an arbitrary sequence of $n$-dimensional conditional distributions $W_n$ from $\mc T^n$ to $\mc Z^n$, where $\mc T$ and $\mc Z$ are the  input and output alphabets, respectively.
There are no other resources available to any of the terminals. 
\begin{definition}
A CR-generation protocol of block-length $n$ consists of:
\begin{enumerate}
    \item A function $\Phi$ that maps $X^n$ into a random variable $K$ with alphabet $\mathcal{K}$ satisfying $\lvert \mc K \rvert \geq 3$ generated by Terminal $A.$
    \item A function $\Lambda$ that maps $X^n$ into the channel input sequence $T^n=(T_1,\hdots,T_n)\in \mc T^n.$
    \item A function $\Psi$ that maps $Y^n$ and the channel output sequence $Z^n=(Z_1,\hdots, Z_n)\in \mc Z^n$ into a random variable $L$ with alphabet $\mathcal{K}$ generated by Terminal $B.$
\end{enumerate}
Such a protocol induces a pair of random variables $(K,L)$ whose joint distribution is determined by $P_{XY}$ and by the channel $\mathbf{W}$. Such a pair of random variables $(K,L)$ is called permissible.
This is illustrated in Fig. \ref{CRprotocol}.
\end{definition}
\begin{figure}
\centering
\tikzstyle{block} = [draw, rectangle, rounded corners,
minimum height=2em, minimum width=2cm]
\tikzstyle{blockchannel} = [draw, top color=white, bottom color=white!80!gray, rectangle, rounded corners,
minimum height=1cm, minimum width=.3cm]
\tikzstyle{input} = [coordinate]
\usetikzlibrary{arrows}
\begin{tikzpicture}[scale= 1,font=\footnotesize]
\node[blockchannel] (source) {$P_{XY}$};
\node[blockchannel, below=2.6cm of source](channel) { Channel $\bs W$};
\node[block, below left=2.2cm of source] (x) {Terminal $A$};
\node[block, below right=2cm of source] (y) {Terminal $B$};
\node[above=1cm of x] (k) {$K=\Phi(X^n)$};
\node[above=1cm of y] (l) {$L=\Psi(Y^n,Z^n)$};

\draw[->,thick] (source) -- node[above] {$X^n$} (x);
\draw[->, thick] (source) -- node[above] {$Y^n$} (y);
\draw [->, thick] (x) |- node[below right] {$T^n=\Lambda(X^n)$} (channel);
\draw[<-, thick] (y) |- node[below left] {$Z^n$} (channel);
\draw[->] (x) -- (k);
\draw[->] (y) -- (l);

\end{tikzpicture}
\caption{Two-source model for CR generation with one-way communication over an arbitrary single-user channel $W.$}
\label{CRprotocol}
\end{figure}
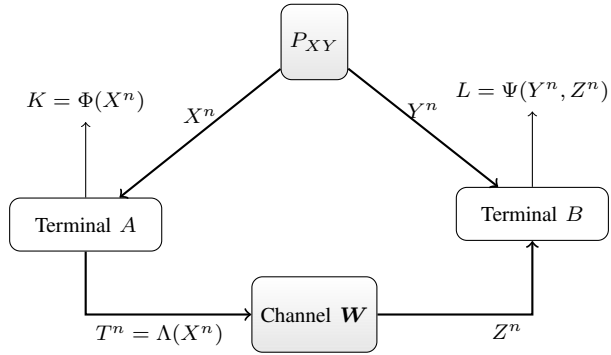
\subsection{Achievable Rate and Capacity}
We define first an achievable $\epsilon$-UCR rate and the $\epsilon$-UCR capacity.
\begin{definition} 
\label{ucrrate}
 Let $0<\epsilon<1.$ A number $H$ is called an achievable $\epsilon$-UCR rate  if there exists a non-negative constant $c$ such that for every $\beta>0,$ $\delta>0$  and for sufficiently large $n$ there exists a permissible  pair of random variables $(K,L)$ such that
\begin{equation}
   \mbb P\left[K\neq L\right]\leq \epsilon, 
    \label{errorcorrelated}
\end{equation}
\begin{equation}
    |\mathcal{K}|\leq 2^{cn},
    \label{cardinalitycorrelated}
\end{equation}
\begin{equation}
\bigg| \frac{1}{n}H(K)-\frac{1}{n}\log \lvert \mc K \rvert \bigg|  \leq \beta,
  \label{uniformity}
\end{equation}
\begin{equation}
    \frac{1}{n}H(K)> H-\delta.
     \label{ratecorrelated}
\end{equation}
\end{definition}
\begin{definition} 
The $\epsilon$-UCR capacity $C_{\epsilon,UCR}(P_{XY},\mathbf{W})$ is the maximum achievable $\epsilon$-UCR rate.
\end{definition}
Next, we define an achievable $\epsilon$-transmission rate and the $\epsilon$-transmission capacity of the channel $\mathbf{W}.$ For this purpose, we begin by providing the definition of a transmission-code for the channel $\mathbf{W}.$
\begin{definition}
\label{defcode}
A transmission-code $\Gamma_n$ of block-length $n$ and size $N_{n}$ for the channel $\mbf W$ is a family of pairs of codewords and decoding regions $\left\{(\bs{t}_\ell,\setd_\ell) \in \mc T^n \times \mc Z^n, \ \ell=1,\ldots,N_{n} \right\}$ such that for all $\ell,j \in \{1,\ldots,N_{n}\}$  
\begin{align}
&\setd_\ell \cap \setd_j = \emptyset,\quad \ell \neq j. \nonumber
\end{align}The maximum error probability is expressed as 
\begin{align}
    e(\Gamma_n)=\underset{\ell \in \{1,\ldots,N_{n}\}}{\max}W_{n}({\setd_\ell^{c}}|\bs{t}_\ell). \nonumber
\end{align}
\end{definition}
\begin{definition}
\label{deftransmissionrate}
  Let $0<\epsilon<1.$ A real number $R$ is called an \textit{achievable} $\epsilon$-transmission rate of the channel $\mathbf{W}$ if for every $\delta>0$ there exists a code sequence $(\Gamma_n)_{n=1}^\infty$, where each code $\Gamma_n$ of block-length $n$ and size $N_n$ is defined according to Definition \ref{defcode},  such that for sufficiently large $n$ 
    \[
        \frac{\log N_{n}}{n}\geq R-\delta
    \]
    and
    \begin{align}
        e(\Gamma_n)\leq\epsilon.
        \nonumber
    \end{align}
\end{definition}
\begin{definition}
The $\epsilon$-transmission capacity of the channel $\mathbf{W}$ is the maximum achievable $\epsilon$-transmission rate for $\mathbf{W}$ and it is denoted by $C_{\epsilon}(\mathbf{W}).$
\end{definition}
A lower and an upper-bound on the $\epsilon$-transmission capacity  were established in \cite{verduhan}.
\begin{theorem}
\label{generalcapacityformula} \cite{verduhan}
Let $0<\epsilon<1.$ The $\epsilon$-transmission capacity $C_{\epsilon}(\mathbf{W})$ satisfies the following:
\begin{align}
C_{\epsilon}(\mathbf{W}) \geq \underset{\bs{T}}{\sup}\sup\{ R: F_{\bs{T}}(R)< \epsilon \} \label{lowerboundtransmissioncapacity}
\end{align}
and
\begin{align}
C_{\epsilon}(\mathbf{W}) \leq \underset{\bs{T}}{\sup}\sup\{ R: F_{\bs{T}}(R) \leq  \epsilon \} \label{upperboundtransmissioncapacity}
\end{align}
where 
\begin{align}
F_{\bs{T}}(R)= \underset{n\rightarrow\infty}{\limsup} \ \mbb P \left[\frac{1}{n}i(T^n;Z^n)\leq R\right], \nonumber
\end{align}
with $\bs{T}$ being an input process in the form of a sequence of finite-dimensional distributions $\bs{T}=\{T^n=(T_1,\hdots,T_n)\}_{n=1}^{\infty}$
and with $\bs{Z}=\{Z^n=(Z_1,\hdots,Z_n)\}_{n=1}^{\infty}$ being the corresponding output sequence of finite-dimensional distributions induced by $\bs{T}$ via the channel $\mbf W,$ 
where for any $(t^n,z^n)\in \mc T^n \times \mc Z^n$
\begin{align}
    i(t^n;z^n)=\log\frac{P_{Z^n,T^n}(z^n,t^n)}{P_{Z^n}(z^n)P_{T^n}(t^n)}. \nonumber
\end{align}
The lower and upper bound in \eqref{lowerboundtransmissioncapacity} and \eqref{upperboundtransmissioncapacity} are equal except possibly at the points of discontinuity of $C_\epsilon(\mbf W),$ of which there are, at most, countably many. 
\end{theorem}
\subsection{Main Result}
\label{mainresult}
In this section, we give an upper and lower bound on the $\epsilon$-UCR capacity for the model presented in Section \ref{systemmodel}.
\begin{theorem}
For the model in Fig \ref{CRprotocol}, the $\epsilon$-UCR capacity $C_{\epsilon,UCR}(P_{XY},\mathbf{W})$ satisfies
\begin{align}
C_{\epsilon,UCR}(P_{XY},\mathbf{W})\geq \underset{ \substack{U \\{\substack{U\circlearrow{X} \circlearrow{Y}\\ I(U;X)-I(U;Y) \leq l(\epsilon)}}}}{\max} I(U;X)
\label{lowerboundepsilonCRcapacity}
\end{align}
and
\begin{align}
C_{\epsilon,UCR}(P_{XY},\mathbf{W})\leq \underset{ \substack{U \\{\substack{U \circlearrow{X} \circlearrow{Y}\\ I(U;X)-I(U;Y) \leq u(\epsilon)}}}}{\max} I(U;X),
\label{upperboundepsilonCRcapacity}
\end{align}
where
\begin{align}
l(\epsilon)=\underset{\bs{T}}{\sup}\sup\{ R: F_{\bs{T}}(R)< \epsilon \}
\label{Retasupell}
\end{align}
and
\begin{align}
u(\epsilon)= \underset{\bs{T}}{\sup}\sup\{ R: F_{\bs{T}}(R) \leq  \epsilon \}.
\label{uepsilon}
\end{align}
with $\bs{T}$ being an input process in the form of a sequence of finite-dimensional distributions $\bs{T}=\{T^n=(T_1,\hdots,T_n)\}_{n=1}^{\infty}$
and with $\bs{Z}=\{Z^n=(Z_1,\hdots,Z_n)\}_{n=1}^{\infty}$ being the corresponding output sequence of finite-dimensional distributions induced by $\bs{T}$ via the channel $W.$ 
\label{main theorem}
The lower and upper bound in \eqref{lowerboundepsilonCRcapacity} and \eqref{upperboundepsilonCRcapacity} hold with equality except at the points where $l(\epsilon)$ and $u(\epsilon)$ do not coincide, of which there are, at most, countably many.
\end{theorem}
\section{Proof of the Lower-bound in Theorem \ref{main theorem} }
\label{directpart}
	We introduce and prove first the following lemma:
	\begin{lemma}
		The function 
	\begin{align*}
	f_{\max}: \epsilon \mapsto \underset{ \substack{U \\{\substack{U\circlearrow{X} \circlearrow{Y}\\ I(U;X)-I(U;Y) \leq l(\epsilon)}}}}{\max} I(U;X),
	\end{align*}
	 where $l(\epsilon)$ is defined in \eqref{Retasupell}, is left-continuous and monotone non-decreasing in $(0,1).$
	\label{nondecreasingfunctionleftcontinuous}
	\end{lemma}
		\begin{proof}
	To prove that $f_{\max}$ is monotone non-decreasing in $(0,1),$ we will show first that the function $l: \epsilon \mapsto \underset{\bs{T}}{\sup}\sup\{ R: F_{\bs{T}}(R)< \epsilon \}$ is non-decreasing in $(0,1),$ where
		\begin{align}
		F_{\bs{T}}(R)= \underset{n\rightarrow\infty}{\limsup} \ \mbb P \left[\frac{1}{n}i(T^n;Z^n)\leq R\right]. \nonumber
		\end{align}
	
		Let $0<\epsilon_{1}<\epsilon_{2}<1.$
		Clearly, for any $\bs{T},$ we have
		\begin{align}
		\{ R: F_{\bs{T}}(R)< \epsilon_{1} \}\subseteq \{ R: F_{\bs{T}}(R)< \epsilon_{2} \}, \nonumber
		\end{align}
		which implies that for any $\bs{T}$
		\begin{align}
		\sup\{ R: F_{\bs{T}}(R)< \epsilon_{1} \} &\leq \sup	\{ R: F_{\bs{T}}(R)< \epsilon_{2} \}. \nonumber \\
		&\leq \underset{\bs{T}}{\sup} \sup	\{ R: F_{\bs{T}}(R)< \epsilon_{2} \}.
		\nonumber \end{align}
		This yields
		\begin{align}
		\underset{\bs{T}}{\sup}\sup\{ R: F_{\bs{T}}(R)< \epsilon_{1} \} \leq \underset{\bs{T}}{\sup} \sup	\{ R: F_{\bs{T}}(R)< \epsilon_{2} \}.
		\nonumber \end{align}
		This shows that the function $l$ is non-decreasing in $(0,1).$
			Now, let $0<\epsilon_{1}<\epsilon_{2}<1.$ It follows that
	$l(\epsilon_{1})\leq l(\epsilon_{2}).$ 
	Therefore, we have
	\begin{align}
	\underset{ \substack{U \\{\substack{U\circlearrow{X} \circlearrow{Y}\\ I(U;X)-I(U;Y) \leq l(\epsilon_1)}}}}{\max} I(U;X) \leq \underset{ \substack{U \\{\substack{U\circlearrow{X} \circlearrow{Y}\\ I(U;X)-I(U;Y) \leq l(\epsilon_2)}}}}{\max} I(U;X). \nonumber
	\end{align}
	This shows that $f_{\max}$ is monotone non-decreasing in $(0,1)$.
Now to prove the left-continuity of $f_{\max}$ in $(0,1),$ it suffices to show that the function $l: \epsilon \mapsto \underset{\bs{T}}{\sup}\sup\{ R: F_{\bs{T}}(R)< \epsilon \}$ is left-continuous in $(0,1).$
Select any $\epsilon \in (0,1)$ and a strictly increasing sequence $(\epsilon_n)_{n=1}^{\infty}$ converging to $\epsilon$ from the left. For any $\bs{T},$ we have
\begin{align}
   \bigcup\limits_{n=1}^{\infty} \bigg\{ R: F_{\bs{T}}(R)<\epsilon_n  \bigg\} = \bigg\{ R: F_{\bs{T}}(R)<\epsilon  \bigg\}. \nonumber 
\end{align}
It follows that
\begin{align}
    \underset{\epsilon_{n}\rightarrow\epsilon}{\lim} \sup \bigg\{ R: F_{\bs{T}}(R)<\epsilon_n \bigg\} = \sup \bigg\{ R: F_{\bs{T}}(R)<\epsilon  \bigg\}. 
\nonumber \end{align}
This yields
\begin{align}
      &\underset{\bs{T}}{\sup} \underset{\epsilon_{n}\rightarrow\epsilon}{\lim} \sup \bigg\{ R: F_{\bs{T}}(R)<\epsilon_n \bigg\}\nonumber \\
      &= \underset{\bs{T}}{\sup}\sup \bigg\{ R: F_{\bs{T}}(R)<\epsilon  \bigg\}. \label{eq1}
\end{align}
We will show now that
\begin{align}
&\underset{\bs T}{\sup}\underset{\epsilon_{n}\rightarrow\epsilon}{\lim} \sup \bigg\{ R: F_{\bs{T}}(R)<\epsilon_n \bigg\} \nonumber \\ 
&\leq \underset{\epsilon_{n}\rightarrow\epsilon}{\lim} \underset{\bs{T}}{\sup}\sup \bigg\{ R: F_{\bs{T}}(R)<\epsilon_n  \bigg\}. \label{eq2}
\end{align}
For any $n$ and any $\bs{T},$ it holds that
 \begin{align}
     \sup\bigg\{ R: F_{\bs{T}}(R)<\epsilon_n \bigg\} \leq \underset{\bs T}{\sup} \sup\bigg\{ R: F_{\bs{T}}(R)<\epsilon_n  \bigg\}.
 \nonumber \end{align}
 Thus, for any $\bs{T},$ we have
  \begin{align}
     &\underset{\epsilon_{n}\rightarrow\epsilon}{\lim} \sup\bigg\{ R: F_{\bs{T}}(R)<\epsilon_n \bigg\} \nonumber \\
     &\leq \underset{\epsilon_{n}\rightarrow\epsilon}{\lim} \underset{\bs T}{\sup} \sup\bigg\{ R: F_{\bs{T}}(R)<\epsilon_ n  \bigg\}.
 \nonumber \end{align}
 This implies \eqref{eq2}.
 
Next, we will show that
 \begin{align}
    &\underset{\epsilon_{n}\rightarrow\epsilon}{\lim} \underset{\bs{T}}{\sup}\sup \bigg\{ R: F_{\bs{T}}(R)<\epsilon_n  \bigg\}  \nonumber \\
    &\leq \underset{\bs{T}}{\sup}\sup \bigg\{ R: F_{\bs{T}}(R)<\epsilon  \bigg\}. \label{eq3}
\end{align}

For any $n$ and any $\bs{T},$ we have 
\begin{align}
\bigg\{ R: F_{\bs{T}}(R)<\epsilon_n  \bigg\} \subseteq \bigg\{ R: F_{\bs{T}}(R)<\epsilon  \bigg\}. \nonumber
\end{align}
Thus, that for any $n$ and any $\bs{T},$ we have
\begin{align}
\sup\bigg\{ R: F_{\bs{T}}(R)<\epsilon_n  \bigg\} &\leq \sup \bigg\{ R: F_{\bs{T}}(R)<\epsilon  \bigg\} \nonumber \\
&\leq \underset{\bs{T}}{\sup}\sup \bigg\{ R: F_{\bs{T}}(R)<\epsilon  \bigg\}.\nonumber
\end{align}
Therefore, for any $n,$
\begin{align}
\underset{\bs{T}}{\sup}\sup\bigg\{ R: F_{\bs{T}}(R)<\epsilon_n  \bigg\} \leq \underset{\bs{T}}{\sup}\sup \bigg\{ R: F_{\bs{T}}(R)<\epsilon  \bigg\}.\nonumber
\end{align}
This implies \eqref{eq3}.

Now, it follows from \eqref{eq1}, \eqref{eq2} and \eqref{eq3} that
\begin{align}
     &\underset{\epsilon_{n}\rightarrow\epsilon}{\lim} \underset{\bs{T}}{\sup}\sup \bigg\{ R: F_{\bs{T}}(R)<\epsilon_n  \bigg\} \nonumber \\ &=\underset{\bs{T}}{\sup}\sup \bigg\{ R: F_{\bs{T}}(R)<\epsilon  \bigg\}. \nonumber
\end{align}
This shows that the function $l$ is left-continuous at $\epsilon$ for any $\epsilon \in (0,1).$ It follows that $f_{\max}$ is left-continuous in $(0,1).$
\end{proof}

Let $0<\epsilon<1.$ It follows from Lemma \ref{nondecreasingfunctionleftcontinuous} that
\begin{align}
&\underset{0<\epsilon'<\epsilon}{\sup} \underset{ \substack{U \\{\substack{U\circlearrow{X} \circlearrow{Y}\\ I(U;X)-I(U;Y) \leq l(\epsilon')}}}}{\max} I(U;X) \nonumber \\
&=\underset{ \substack{U \\{\substack{U\circlearrow{X} \circlearrow{Y}\\ I(U;X)-I(U;Y) \leq l(\epsilon)}}}}{\max} I(U;X).
\label{epsilonprimeepsilon}
\end{align}
Let $0<\epsilon'<\epsilon$ be fixed arbitrarily. From \eqref{epsilonprimeepsilon}, it suffices to show that $$\underset{ \substack{U \\{\substack{U\circlearrow{X} \circlearrow{Y}\\ I(U;X)-I(U;Y) \leq l(\epsilon')}}}}{\max} I(U;X)$$ is an achievable $\epsilon$-UCR rate.
\subsubsection{If $l(\epsilon')=0$}  It is shown in \cite{part2} that when the terminals do not communicate over the channel, the UCR capacity is equal to 
\begin{align}
H_{0}=\underset{ \substack{U \\{\substack{U \circlearrow{X} \circlearrow{Y}\\ I(U;X)-I(U;Y) \leq 0}}}}{\max} I(U;X).  \nonumber
\end{align}
Hence, when the terminals do not communicate over the channel $\mbf W$,  $H_0$ is also an achievable $\epsilon$-UCR rate

\subsubsection{If $ l(\epsilon')>0$} We extend the CR generation scheme provided in \cite{part2} to arbitrary single-user channels. By continuity, it suffices to show that 
$$ \underset{ \substack{U \\{\substack{U \circlearrow{X} \circlearrow{Y}\\ I(U;X)-I(U;Y) \leq C'}}}}{\max} I(U;X)  $$ is an achievable  $\epsilon$-UCR rate for every $C'<l(\epsilon').$
Let $U$ be any random variable with alphabet $\mc U$ satisfying $U \circlearrow{X} \circlearrow{Y}$ and $I(U;X)-I(U;Y) \leq C'$. Let $\delta,\beta>0.$ We are going to show that $H=I(U;X)$ is an achievable $\epsilon$-UCR rate. Without loss of generality, assume that the distribution of $U$ is a possible type for block-length $n$.
For some $\mu>0,$ we let
{{\begin{align}
N_{1}&=\lfloor 2^{n[I(U;X)-I(U;Y)+3\mu]} \rfloor \nonumber\\
N_{2}&=\lfloor 2^{n[I(U;Y)-2\mu]}\rfloor. \nonumber
\end{align}}}For each pair $(i,j)$ with $1\leq i \leq N_{1}$ and $1\leq j \leq N_{2}$, we define a random sequence $\bs{U}_{i,j}\in\mathcal{U}^n$ of type $P_{U}$. Let $\mbf M=\bs{U}_{1,1},\hdots, \bs{U}_{N_{1},N_{2}}$  be the joint random variable of all $\bs{U}_{i,j}s.$  We define $\Phi_{\mbf M}$ as follows:
 Let $\Phi_{\mbf M}(X^n)=\bs{U}_{ij}$, if $\bs{U}_{ij}$ is jointly $UX$-typical with $X^n$ (either one if there are several). If no such $\bs{U}_{i,j}$ exists, then  $\Phi_{\mbf M}(X^n)$ is set to a constant sequence $\bs{u}_0$ different from all the  ${\bs{U}_{ij}}s$, jointly $UX$-typical with none of the realizations of $X^n$ and known to both terminals. We further define the following two sets
\begin{align}
    S_{1}(\mbf M)&=\{(x^{n},y^{n}):(\Phi_{\mbf M}(x^{n}),x^{n},y^{n}) \in \mathcal{T}_{U,X,Y}^{n}\} \nonumber
\end{align} and
\begin{align}
    &S_{2}(\mbf M) \nonumber \\
    &=\Big\{(x^{n},y^{n}):(x^{n},y^{n}) \in S_{1}(\mbf M) \ \text{s.t.} \ \bs{U}_{i,j}=\Phi_{\mbf M}(x^{n}) \nonumber \\   & \ \ \ \  \ \text{and} \ \exists \ \bs{U}_{i,\ell}\neq\bs{U}_{i,j} \ \text{jointly} \ UY\text{-typical with} \ y^{n} \nonumber \\
    & \ \ \ \ \ (\text{with the same first index} \ i)
\Big\}.\nonumber
\end{align}
It is proved in \cite{part2} that
\begin{align}
    \mathbb{E}_{\mbf M}\left[ \mbb P\left[(X^n,Y^n)\notin  S_{1}(\mbf M)\right]+\mbb P\left[(X^n,Y^n)\in  S_{2}(\mbf M)\right]\right]\leq \zeta(n),
    \label{averagebeta}
\end{align}
where $\zeta(n) \leq \epsilon-\epsilon'$ for sufficiently large $n$. 
We choose a realization $\mbf m=\bs{u}_{1,1},\hdots, \bs{u}_{N_1,N_2}$ satisfying:
\begin{align}
\mbb P\left[(X^n,Y^n)\notin  S_{1}(\mbf m)\right]+\mbb P\left[(X^n,Y^n)\in  S_{2}(\mbf m)\right]\leq \zeta(n).
\label{choiceofm} \end{align} 
From \eqref{averagebeta}, we know that such a realization exists. We denote $\Phi_{\mbf m}$ by $\Phi.$
We assume that each $\bs{u}_{i,j}, \ i=1\hdots N_1, \ j=1\hdots N_2,$  is known to both terminals.  
This means that  $N_{1}$ codebooks $C_{i}, 1\leq i \leq N_{1}$, are known to both terminals, where each codebook contains $N_{2}$ sequences, $ \bs{u}_{i,j}, \ j=1,\hdots, N_2$. 

Let $x^{n}$ be any realization of $X^n$ and $y^{n}$ be any realization of $Y^n.$  Let $f_1(x^n)=i$ if $\Phi(x^n)=\bs{u}_{i,j}$. Otherwise, if $\Phi(x^n)=\bs{u}_{0},$ then $f_1(x^n)=N_1+1.$.

  Since $ C'<l(\epsilon')$, we choose $\mu$ to be sufficiently small such that
      \begin{align}
     \frac{\log(N_1+1)}{n} 
     &\leq l(\epsilon')-\mu',
     \label{inequalitylogfSISO}
      \end{align}
for some $\mu'>0.$
 The message $i^\star=f_1(x^{n})$, with $i^\star\in\{1,\hdots,N_1+1\}$, is encoded to a sequence $t^n$ using a code sequence $(\Gamma^\star_n)_{n=1}^{\infty}$ with rate $\frac{\log(N_1+1)}{n} $ satisfying \eqref{inequalitylogfSISO}
 and with error probability $e(\Gamma^\star_n)$ satisfying $e(\Gamma^\star_n) \leq \epsilon',$ for sufficiently large $n.$
  From the definition of an achievable $\epsilon'$- transmission rate, we know that such a code sequence exists. The sequence $t^n$ is sent over the single-user channel $W_n$. Let $z^n$ be the channel output sequence. Terminal $B$ decodes the message $\tilde{i}^\star$ from the knowledge of $z^n.$
Let $\Psi(y^{n},z^n)=\bs{u}_{\tilde{i}^\star,j}$ if $\bs{u}_{\tilde{i}^\star,j}$ and $y^{n}$ are jointly $UY$-typical . If there is no such  $\bs{u}_{\tilde{i}^\star,j}$ or there are several, we set $\Psi(y^{n},z^n)=\bs{u}_0.$ 

For $c=I(U;X)+\mu+1,$  we have $\lvert \mc K \rvert = N_1 N_2+1 \leq 2^{nc}.$
We define for any $(i,j)\in \{1,\hdots,N_{1}\}\times\{1,\hdots,N_{2}\}$  the set
$$\mc R=\{ x^{n}\in\mathcal{X}^{n} \ \text{s.t.} \ (\bs{u}_{i,j},x^{n}) \ \text{jointly} \ UX\text{-typical}\}.$$
Then, it holds that 
\begin{align}
\mbb P[K=\bs{u}_{i,j}] &\overset{(a)}{=}\sum_{x^{n}\in\mc R}\mbb P[K=\bs{u}_{i,j}|X^n=x^{n}]P_{X}^n(x^{n}) \nonumber \\
&\leq \sum_{x^{n}\in\mc R}P_{X}^n(x^{n}) \nonumber \\
&=P_{X}^{n}(\{x^{n}: (\bs{u}_{i,j},x^{n}) \ \text{jointly} \ UX\text{-typical}\}) \nonumber\\
& = 2^{-nI(U;X)-\kappa(n)}, \nonumber
\end{align}
for some $\kappa(n)>0$ with $\underset{n\rightarrow \infty}{\lim} \frac{\kappa(n)}{n}=0$,
where $(a)$ follows because for  $(\bs{u}_{i,j},\mathbf{x})$ being not jointly $UX$-typical, we have $\mbb P[K=\bs{u}_{i,j}|X^n=x^{n}]=0.$ This yields
{{\begin{align}
H(K)\geq nI(U;X)-\kappa'(n)
\nonumber \end{align}}}
for some $\kappa'(n)>0$ with $\underset{n\rightarrow \infty}{\lim} \frac{\kappa'(n)}{n}=0.$

Therefore, for sufficiently large $n,$ it holds that
$\frac{H(K)}{n}>H-\delta.$
Clearly, it holds also that $\frac{1}{n}\bigg\lvert H(K)-\log\lvert \mc K \rvert \bigg\rvert \leq \kappa''(n) $
for some $\kappa''(n)>0$ with $\underset{n\rightarrow\infty}{\lim} \kappa''(n)=0.$ Therefore, for sufficiently large $n,$ it holds that $\kappa''(n)\leq \beta.$ 
 Let $I^\star=f_1(X^n)$ be the random message generated by Terminal $A$ and  $\tilde{I}^\star$  be the random message decoded by Terminal $B$. 

We have
\begin{align}
    \mbb P[K\neq L]
    &=\mbb P[K\neq L|I^\star=\tilde{I}^\star]\mbb P[I^\star=\tilde{I}^\star] \nonumber \\  
    &\quad+ \mbb P[K\neq L|I^\star\neq \tilde{I}^\star]\mbb P[I^\star\neq\tilde{I}^\star] \nonumber \\
        &\leq \mbb P[K\neq L|I^\star=\tilde{I}^\star]+ \mbb P[I^\star\neq\tilde{I}^\star].\nonumber
\end{align}
Let $\mathcal{D}_{\mbf m}= ``\Phi(X^n) \ \text{is equal to none of the} \  {\bs{u}_{i,j}}'s".$
We denote its complement by $\mc D_{\mbf m}^{c}.$
It holds that
\begin{align}
    &\mbb P[K\neq L|I^\star=\tilde{I}^\star] \nonumber \\
   &\overset{(a)}{=}\mbb P[K\neq L|I^\star=\tilde{I}^\star,\mathcal{D}_{\mbf m}^c]\mbb P[\mathcal{D}_{\mbf m}^c|I^\star=\tilde{I}^\star] \nonumber \\
   &\leq \mbb P[K\neq L|I^\star=\tilde{I}^\star,\mathcal{D}_{\mbf m}^c],\nonumber
\end{align}
where $(a)$ follows from $\mbb P[K\neq L|I^\star=\tilde{I}^\star,\mathcal{D}_{\mbf m}]=0,$ since conditioned on  $I^\star=\tilde{I}^\star$ and $\mathcal{D}_{\mbf m}$, we know that $K$ and $L$ are both equal to $\bs{u}_0$.
It follows that
\begin{align}
    &\mbb P[K\neq L] \nonumber \\
    &\leq \mbb P[K\neq L|I^\star=\tilde{I}^\star,\mathcal{D}_{\mbf m}^c]+ \mbb P[I^\star\neq\tilde{I}^\star] \nonumber \\
    &\leq \mbb P\left[(X^n,Y^n)\in  S_{1}^{c}(\mbf m)\cup S_{2}(\mbf m)\right]+\mbb P[I^\star\neq\tilde{I}^\star] \nonumber \\
    &\overset{(a)}{=}\mbb P\left[(X^n,Y^n)\notin  S_{1}(\mbf m)\right]+\mbb P\left[(X^n,Y^n)\in  S_{2}(\mbf m)\right] +\mbb P[I^\star\neq\tilde{I}^\star] \label{upperboundzwischenschritt} 
\end{align}
where $(a)$ follows because $S_{1}^{c}(\mbf m)$ and $S_{2}(\mbf m)$ are disjoint. 
 
It follows from \eqref{upperboundzwischenschritt} using \eqref{choiceofm} that
\begin{align}
        \mbb P[K\neq L] &\leq  \zeta(n)+ \mbb P[I^\star\neq\tilde{I}^\star],\nonumber\\
    &\overset{(a)}{\leq} \epsilon-\epsilon'+\epsilon' \nonumber \\
    &=\epsilon. \nonumber
\end{align}
where $(a)$ follows because $\zeta(n)\leq \epsilon-\epsilon'$ and $e(\Gamma^\star_n) \leq \epsilon'.$
This completes the proof of the lower-bound on the $\epsilon$-UCR capacity.
\section{Proof of the upper-bound in Theorem \ref{main theorem}}
\label{converseproof}

 Let $0<\epsilon<1.$
Let $H$ be any achievable $\epsilon$-UCR rate.  So, there exists a non-negative constant $c$ such that for every $\delta,\beta>0$ and for sufficiently large $n,$ there exists a permissible pair of random variables $(K,L)$ according to a fixed CR-generation protocol of block-length $n$ such that \eqref{errorcorrelated}, \eqref{cardinalitycorrelated}, \eqref{uniformity} and \eqref{ratecorrelated} are satisfied.
Define $\lambda(\beta)=\beta+ 2\beta c+\beta^2.$ 
Let $$ \gamma(\epsilon,\beta)=2\sqrt{\frac{\sqrt{\lambda(\beta)}}{1-\sqrt{\epsilon}}}$$ and
$$
    \kappa(\epsilon,\beta)=\epsilon+1-\left(1-4\frac{\lambda(\beta)}{\gamma(\epsilon,\beta)^2}\right)^2. $$
Define $\mathcal{B}_1=\{\beta:0<\beta<1 \ \text{and} \ \epsilon<\kappa(\epsilon,\beta)+\beta <1\}$ and $\mathcal{B}_{2}=\{\beta: 0<\lambda(\beta)<1\}.$ Let $\mc B=\mathcal{B}_1 \cap \mathcal{B}_{2}. $ The set $\mathcal{B}$ is clearly  non-empty since any sufficiently small $\beta>0$ is element of $\mathcal{B}.$ 
Assume without loss of generality the constant $\beta>0$ in \eqref{uniformity} is element of $\mathcal{B}.$ Define
\begin{align}
&R_{\epsilon,\sup} \nonumber \\
&= \sup\left\{R: \underset{n\rightarrow\infty}{\limsup} \ \mbb P\left[\frac{1}{n}\log\frac{P_{Z^n,T^n}(Z^n,T^n)}{P_{Z^n}(Z^n)P_{T^n}(T^n)}\leq R \right]\leq \epsilon  \right\}\nonumber \end{align}
and notice that $R_{\epsilon,\sup}\leq u(\epsilon).$
Define 
\begin{align}
&\mc E \nonumber \\& = \begin{aligned}[t]\Biggl\{&\mu>0\ \text{s.t.} \ \text{for infinitely many} \ n:  \\  &\kappa(\epsilon,\beta)+\beta<\mbb P\left[\frac{\log\frac{P_{Z^n,T^n}(Z^n,T^n)}{P_{Z^n}(Z^n)P_{T^n}(T^n)}}{n} \leq R_{\epsilon,\sup}+\mu\right] \Biggl\}.\end{aligned} \nonumber
\end{align}

From the definition of $R_{\epsilon,\sup}$ and since $\epsilon<\kappa(\epsilon,\beta)+\beta<1,$  we know that $\mc E$ is a non-emtpy set.
Let $\mu$ be an arbitrary element of $\mc E.$ For infinitely many $n,$ it holds that

 \begin{align}
\kappa(\epsilon,\beta)+\beta<\mbb P\left[\frac{1}{n}\log\frac{P_{Z^n,T^n}(Z^n,T^n)}{P_{Z^n}(Z^n)P_{T^n}(T^n)} \leq R_{\epsilon,\sup}+\mu \right]. \label{choiceepsilon}
 \end{align}
 \begin{claim}
 For sufficiently large $n$ satisfying \eqref{choiceepsilon}, it holds that
 \label{claim1}
\begin{align}
    \frac{H(K|Y^n)}{n}\leq u(\epsilon)+\zeta(n,\epsilon,\beta,\mu), \nonumber
 \end{align}
 where  $u(\epsilon)$ is defined in \eqref{uepsilon} and where $\zeta(n,\epsilon,\beta,\mu)=\mu+\gamma(\epsilon,\beta)+\frac{2}{n}\log\frac{1}{\beta}.$
 \end{claim} 
In order to prove the claim, we will use a change of measure argument. To prepare this, we need some technicalities. 
Let $\mc B=\mc K\times \mc Y^n.$ Consider now the set
\begin{align}
    &\mc D \nonumber \\
    &=\Big\{ (k,y^n)\in \mc B:  \frac{\log\frac{1}{P_{K|Y^n}(k|y^n)}}{n} \geq \frac{H(K|Y^n)}{n}-\gamma(\epsilon,\beta) \Big\}.
\nonumber \end{align}
Let $\mc A=\mc X^{n}\times \mc Y^{n} \times \mc Z^{n}.$
Define the sets 
\begin{align}
     &\mc S_{1} \nonumber \\
     &=\Big\{ (x^n,y^n,z^n)\in \mc A: \frac{ \log\frac{P_{Z^n|T^n}(z^n|\Lambda(x^n))}{P_{Z^n}(z^n)}}{n}  \leq R_{\epsilon,\sup}+\mu\Big\},\nonumber
\end{align}
\begin{align}
    \mc S_{2}=\{(x^n,y^n,z^n) \in \mc A: \Phi(x^n)=\Psi(y^n,z^n)   \}, \nonumber
\end{align}
and
\begin{align}
    \mc S_{3}=\{(x^n,y^n,z^n)\in \mc A: (\Phi(x^n),y^n)\in \mc D\}. \nonumber
\end{align}

Let $\mc S=\mc S_{1} \cap \mc S_{2} \cap \mc S_{3}.$  We introduce now and prove the following lemma. \begin{lemma} 
For sufficiently large $n$ satisfying \eqref{choiceepsilon}, we have
\label{probS}
\begin{align*}
   \mbb P \left[ (X^n,Y^n,Z^n)\in \mc S\right]\geq \beta>0. \nonumber  
\end{align*}
\end{lemma}
Analogously to \cite{strongconverse}, we change the probability measure by defining 
\begin{align}
  &P_{\tilde{X}^n,\tilde{Y}^n,\tilde{Z}^n}(x^n,y^n,z^n) \nonumber \\
  &=\frac{P_{X^n,Y^n,Z^n}(x^n,y^n,z^n) \mbf 1 \left[(x^n,y^n,z^n)\in \mc S  \right]}{\mbb P \left[(X^n,Y^n,Z^n) \in \mc S\right]}, \nonumber
\end{align}
where $\mbf 1[\cdot]$ is the indicator function.

\begin{proof}
It holds for $   \kappa(\epsilon,\beta)=\epsilon+\left[1-\left(1-4\frac{\lambda(\beta)}{\gamma(\epsilon,\beta)^2}\right)^2\right]$ that
\begin{align}
    &\mbb P\left[ (X^n,Y^n,Z^n) \in \mc S\right] \nonumber \\
    &\geq 1- \mbb P\left[(X^n,Y^n) \notin \mc S_{3}\right]-\mbb P\left[(X^n,Y^n,Z^n) \notin \mc S_{2}\right]  \nonumber \\
    &\quad- \mbb P\left[(X^n,Y^n,Z^n)\notin \mc S_{1}\right]
    \nonumber \\
   &= 1-\mbb P\left[(K,Y^n)\notin\mc D\right]-\mbb P\left[K\neq L\right] -\mbb P\left[(X^n,Y^n,Z^n)\notin \mc S_{1}\right]  \nonumber \\
   &\overset{(a)}{\geq} 1-\left[1-\left(1-4\frac{\lambda(\beta)}{\gamma(\epsilon,\beta)^2}\right)^2\right]-\epsilon -\mbb P\left[(X^n,Y^n,Z^n)\notin \mc S_{1}\right] \nonumber \\
   &=1-\kappa(\epsilon,\beta)-\mbb P\left[\frac{1}{n}\log\frac{P_{Z^n,T^n}(Z^n,T^n)}{P_{Z^n}(Z^n)P_{T^n}(T^n)} > R_{\epsilon,\sup}+\mu \right] \nonumber \\
   &\overset{(b)}{\geq}  1-\kappa(\epsilon,\beta)-\left( 1-\kappa(\epsilon,\beta)-\beta\right)                    \nonumber\\
   &=\beta, \nonumber
   \end{align}
where $(a)$ follows from Lemma \ref{boundprobsetd} in the appendix and $(b)$ follows from the choice of $\mu$ in \eqref{choiceepsilon}.
\end{proof}

From Lemma \ref{probS}, we know that  $(\tilde{X}^n,\tilde{Y}^n,\tilde{Z}^n)$ is well-defined.

Let $\tilde{K}=\Phi(\tilde{X}^n).$ Let $\tilde{T}^n=\Lambda(\tilde{X}^n)$ and
 $\tilde{L}=\Psi(\tilde{Y}^n,\tilde{Z}^n).$ Here, $\tilde{K}$ is equal to $\tilde{L}$ with probability one. Furthermore, for every $(x^n,z^n)\in \text{supp}(\tilde{X}^n)\times \text{supp}(\tilde{Z}^n),$ we have 
\begin{align}
      \frac{1}{n} \log\frac{P_{T^n,Z^n}(\Lambda(x^n),z^n)}{P_{T^n}(\Lambda(x^n))P_{Z^n}(z^n)} \leq R_{\epsilon,\sup}+\mu. \nonumber
\end{align}
It follows that
\begin{align}
    \frac{1}{n}\mbb E\left[\log\frac{P_{T^nZ^n}(\Lambda(\tilde{X}^n),\tilde{Z}^n)}{P_{T^n}(\Lambda(\tilde{X}^n))P_{Z^n}(\tilde{Z}^n)}\right]\leq R_{\epsilon,\sup}+\mu. \label{upperboundnewmutinf} \end{align}
    Let us now introduce the following two lemmas 
\begin{lemma} It holds that
\label{upperboundlogcardinality}
\begin{align}
    H(K|Y^n) &\leq n\gamma(\epsilon,\beta)+\log\frac{1}{\beta}+H(\tilde{K}|\tilde{Y}^n). \nonumber
\end{align}
\end{lemma}
\begin{proof}
Consider any $(k,y^n) \in \text{supp}(\tilde{K})\times \text{supp}(\tilde{Y}^n).$ 
If $(k,y^n) \not \in \mc D,$ it holds that

$$\frac{P_{\tilde{K},\tilde{Y}^n}(k,y^n)}{P_{Y^n}(y^n)}=0.$$

Now, if $(k,y^n) \in \mc D,$ we have using Lemma \ref{probS}

\begin{align}
   &\frac{P_{\tilde{K},\tilde{Y}^n}(k,y^n)}{P_{Y^n}(y^n)}\nonumber\\
    &=\frac{1}{P_{Y^n}(y^n)}\sum_{\substack{x^n,z^n \\(x^n,y^n,z^n)\in \mc S \\ \Phi(x^n)=k}} P_{\tilde{X}^n,\tilde{Y}^n,\tilde{Z}^n}(x^n,y^n,z^n)\nonumber \\
    &=\frac{1}{P_{Y^n}(y^n)}\sum_{\substack{x^n,z^n \\(x^n,y^n,z^n)\in \mc S \\ \Phi(x^n)=k}} \frac{P_{X^n,Y^n,Z^n}(x^n,y^n,z^n)}{\mbb P\left[(X^n,Y^n,Z^n)\in \mc S\right]} \nonumber \\
      &\leq \frac{P_{K,Y^n}(k,y^n)}{P_{Y^n}(y^n)\mbb P\left[(X^n,Y^n,Z^n)\in \mc S\right]} \nonumber \\
       &= \frac{P_{K|Y^n}(k|y^n)}{\mbb P\left[(X^n,Y^n,Z^n)\in \mc S\right]}  \nonumber\\
         &\overset{(a)}{\leq} \frac{2^{n\gamma(\epsilon,\beta)}}{2^{H(K|Y^n)} \mbb P\left[(X^n,Y^n,Z^n)\in \mc S\right]} \nonumber \\
         &\leq  \frac{2^{n\gamma(\epsilon,\beta)}}{2^{H(K|Y^n)} \beta} , \nonumber
         \end{align}
where $(a)$ follows because $(k,y^n)\in \mc D.$ 
Therefore, for every $(k,y^n) \in \text{supp}(P_{\tilde{K},\tilde{Y}^n}),$ we have
\begin{align}
	    \frac{P_{\tilde{K},\tilde{Y}^n}(k,y^n)}{P_{Y^n}(y^n)} \leq \frac{2^{n\gamma(\epsilon,\beta)}}{2^{H(K|Y^n)} \beta}, \nonumber
\end{align}
which yields
\begin{align}
    2^{H(K|Y^n)} \leq \frac{2^{n\gamma(\epsilon,\beta)}}{\beta} \frac{1}{\frac{P_{\tilde{K},\tilde{Y}^n}(k,y^n)}{P_{Y^n}(y^n)}}. \nonumber
\end{align}
This implies that for any $(k,y^n) \in \text{supp}(P_{\tilde{K},\tilde{Y}^n})$ 
\begin{align}
H(K|Y^n)&\leq \log\frac{2^{n\gamma(\epsilon,\beta)}}{\beta}-\log\frac{P_{\tilde{K},\tilde{Y}^n}(k,y^n)}{P_{Y^n}(y^n)}. \nonumber 
\end{align}
As a result, it follows that
\begin{align}
 H(K|Y^n) &\leq \log\frac{2^{n\gamma(\epsilon,\beta)}}{\beta} \nonumber \\
 &\quad+\underset{(k,y^n)\in \text{supp}(P_{\tilde{K},\tilde{Y}^n})}{\min}-\log\frac{P_{\tilde{K},\tilde{Y}^n}(k,y^n)}{P_{Y^n}(y^n)}. \nonumber   
\end{align}
Now, it holds that
\begin{align}
    &\underset{(k,y^n)\in \text{supp}(P_{\tilde{K},\tilde{Y}^n})}{\min}-\log\frac{P_{\tilde{K},\tilde{Y}^n}(k,y^n)}{P_{Y^n}(y^n)} \nonumber \\ &\leq \mbb E\left[-\log\frac{P_{\tilde{K},\tilde{Y}^n}(\tilde{K},\tilde{Y}^n)}{P_{Y^n}(\tilde{Y}^n)}\right] \nonumber \\ 
    &=\mbb E\left[-\log P_{\tilde{K}|\tilde{Y}^n}(\tilde{K}|\tilde{Y}^n) \right]-\mbb E\left[\log\frac{P_{\tilde{Y}^n}(\tilde{Y}^n)}{P_{Y^n}(\tilde{Y}^n)}\right]\nonumber \\
    &=H(\tilde{K}|\tilde{Y}^n)-D(P_{\tilde{Y}^n}||P_{Y^n}) \nonumber \\ \nonumber \\
    &\leq H(\tilde{K}|\tilde{Y}^n). \nonumber
\end{align}
 It follows that
\begin{align}
H(K|Y^n) &\leq \log\frac{2^{n\gamma(\epsilon,\beta)}}{\beta}+H(\tilde{K}|\tilde{Y}^n) \nonumber \\
&=n\gamma(\epsilon,\beta)+\log\frac{1}{\beta}+H(\tilde{K}|\tilde{Y}^n).\nonumber
\end{align}
\end{proof}
\begin{lemma} It holds that
\label{mutinftilektildezcondtildey}
\begin{align}
    \frac{1}{n} I(\tilde{K};\tilde{Z}^{n}|\tilde{Y}^{n})\leq R_{\epsilon,\sup}+\mu+\frac{1}{n}\log\frac{1}{\beta}. \nonumber
\end{align}
\end{lemma}
\begin{proof}
We have
\begin{align} 
&\frac{1}{n} I(\tilde{K};\tilde{Z}^{n}|\tilde{Y}^{n}) \nonumber \\ \nonumber\\
&\leq \frac{1}{n} I(\tilde{X}^{n}\tilde{K};\tilde{Z}^{n}|\tilde{Y}^{n}) \nonumber \\ \nonumber \\
&=\frac{1}{n}\mbb E \left[\log\frac{P_{\tilde{Y}^n}(\tilde{Y}^n)P_{\tilde{X}^n,\tilde{K},\tilde{Z}^{n},\tilde{Y}^n}(\tilde{X}^n,\tilde{K},\tilde{Z}^{n},\tilde{Y}^n)}{P_{\tilde{X}^n,\tilde{K},\tilde{Y}^n}(\tilde{X}^n,\tilde{K},\tilde{Y}^n)P_{\tilde{Z}^n,\tilde{Y}^n}(\tilde{Z}^n,\tilde{Y}^n)}\right] \nonumber \\ \nonumber \\
&=\frac{1}{n}\mbb E\left[\log\frac{P_{T^n,Z^n}(\Lambda(\tilde{X}^n),\tilde{Z}^n)}{P_{T^n}(\Lambda(\tilde{X}^n))P_{Z^n}(\tilde{Z}^n)}\right] \nonumber \\
&\quad+\frac{1}{n}\mbb E \left[\log\frac{P_{\tilde{X}^n,\tilde{K},\tilde{Z}^{n},\tilde{Y}^n}(\tilde{X}^n,\tilde{K},\tilde{Z}^{n},\tilde{Y}^n)}{P_{\tilde{X}^n,\tilde{K},\tilde{Y}^n}(\tilde{X}^n,\tilde{K},\tilde{Y}^n)P_{Z^n|T^n}(\tilde{Z}^n|\Lambda(\tilde{X}^n))}\right] \nonumber \\&\quad-\frac{1}{n}\mbb E\left[\log\frac{P_{\tilde{Z}^n,\tilde{Y}^n}(\tilde{Z}^n,\tilde{Y}^n)}{P_{\tilde{Y}^n}(\tilde{Y}^n)P_{Z^n}(\tilde{Z}^n)} \right] \nonumber \\ \nonumber \\
&=\frac{1}{n}\mbb E\left[\log\frac{P_{T^n,Z^n}(\Lambda(\tilde{X}^n),\tilde{Z}^n)}{P_{T^n}(\Lambda(\tilde{X}^n))P_{Z^n}(\tilde{Z}^n)}\right] \nonumber \\
&\quad+\frac{1}{n}\mbb E \left[\log\frac{P_{\tilde{X}^n,\tilde{K},\tilde{Z}^{n},\tilde{Y}^n}(\tilde{X}^n,\tilde{K},\tilde{Z}^{n},\tilde{Y}^n)}{P_{\tilde{X}^n,\tilde{K},\tilde{Y}^n}(\tilde{X}^n,\tilde{K},\tilde{Y}^n)P_{Z^n|T^n}(\tilde{Z}^n|\Lambda(\tilde{X}^n))}\right] \nonumber \\
&\quad-\frac{1}{n}D(P_{\tilde{Z}^n,\tilde{Y}^n}||P_{\tilde{Y}^n}P_{Z^n}) \nonumber \\ \nonumber \\
&\leq \frac{1}{n}\mbb E\left[\log\frac{P_{T^n,Z^n}(\Lambda(\tilde{X}^n),\tilde{Z}^n)}{P_{T^n}(\Lambda(\tilde{X}^n))P_{Z^n}(\tilde{Z}^n)}\right] \nonumber \\
&\quad+\frac{1}{n}\mbb E \left[\log\frac{P_{\tilde{X}^n,\tilde{K},\tilde{Z}^{n},\tilde{Y}^n}(\tilde{X}^n,\tilde{K},\tilde{Z}^{n},\tilde{Y}^n)}{P_{\tilde{X}^n,\tilde{K},\tilde{Y}^n}(\tilde{X}^n,\tilde{K},\tilde{Y}^n)P_{Z^n|T^n}(\tilde{Z}^n|\Lambda(\tilde{X}^n))}\right].
\nonumber \end{align}

Now, consider any $(x^n,k,z^n,y^n)\in \text{supp}(\tilde{X}^n)\times \text{supp}(\tilde{K}) \times \text{supp}(\tilde{Z}^n)\times \text{supp}(\tilde{Y}^n).$ 
If $\Phi(x^n)\neq k,$ then we have
$$P_{\tilde{X}^n,\tilde{K},\tilde{Z}^{n},\tilde{Y}^n}(x^n,k,z^n,y^n)=P_{X^n,K,Y^n}(x^n,k,y^n)     =0.$$
If $\Phi(x^n)=k,$ then we have
\begin{align}
&P_{\tilde{X}^n,\tilde{K},\tilde{Z}^{n},\tilde{Y}^n}(x^n,k,z^n,y^n) \nonumber \\
&= \frac{1}{\mbb P\left[ (X^n,Y^n,Z^n) \in \mc S\right]} P_{X^n,K,Z^{n},Y^n}(x^n,k,z^n,y^n) \nonumber \\
&\overset{(a)}{\leq} \frac{1}{\beta}P_{X^n,K,Z^{n},Y^n}(x^n,k,z^n,y^n) \nonumber \\
&=\frac{1}{\beta} P_{Z^{n}|X^n,K,Y^n}(z^n|x^n,k,y^n) P_{X^n,K,Y^n}(x^n,k,y^n) \nonumber \\
&\overset{(b)}{=}\frac{P_{Z^{n}|X^n,T^n,K,Y^n}(z^n|x^n,\Lambda(x^n),k,y^n) P_{X^n,K,Y^n}(x^n,k,y^n)}{\beta}  \nonumber \\
&\overset{(c)}{=} \frac{1}{\beta} P_{Z^{n}|T^n}(z^n|\Lambda(x^n)) P_{X^n,K,Y^n}(x^n,k,y^n), \nonumber 
\end{align}
where $(a)$ follows from Lemma \ref{probS}, $(b)$ follows because $T^n=\Lambda(X^n)$ and $(c)$ follows because $Y^{n}\circlearrow{X^{n}K}\circlearrow{T^n}\circlearrow{Z^{n}}$ forms a Markov chain.

Therefore, for any  $(x^n,k,z^n,y^n)\in \text{supp}(\tilde{X}^n)\times \text{supp}(\tilde{K}) \times \text{supp}(\tilde{Z}^n)\times \text{supp}(\tilde{Y}^n),$ we have
\begin{align}
&P_{\tilde{X}^n,\tilde{K},\tilde{Z}^{n},\tilde{Y}^n}(x^n,k,z^n,y^n)	\nonumber \\
&\leq \frac{1}{\beta} P_{Z^{n}|T^n}(z^n|\Lambda(x^n)) P_{X^n,K,Y^n}(x^n,k,y^n).
\nonumber \end{align}
This implies that for any $(x^n,k,z^n,y^n)\in \text{supp}(\tilde{X}^n)\times \text{supp}(\tilde{K}) \times \text{supp}(\tilde{Z}^n)\times \text{supp}(\tilde{Y}^n),$ we have
\begin{align}
&\frac{P_{\tilde{X}^n,\tilde{K},\tilde{Z}^{n},\tilde{Y}^n}(x^n,k,z^n,y^n)}{P_{\tilde{X}^n,\tilde{K},\tilde{Y}^n}(x^n,k,y^n)P_{Z^n|T^n}(z^n|\Lambda(x^n))}\nonumber \\ &\leq \frac{1}{\beta}\frac{P_{X^n,K,Y^n}(x^n,k,y^n)}{P_{\tilde{X}^n,\tilde{K},\tilde{Y}^n}(x^n,k,y^n)}.\nonumber 
\end{align}
Thus, we have
\begin{align}
	&\mbb E \left[\log\frac{P_{\tilde{X}^n,\tilde{K},\tilde{Z}^{n},\tilde{Y}^n}(\tilde{X}^n,\tilde{K},\tilde{Z}^{n},\tilde{Y}^n)}{P_{\tilde{X}^n,\tilde{K},\tilde{Y}^n}(\tilde{X}^n,\tilde{K},\tilde{Y}^n)P_{Z^n|T^n}(\tilde{Z}^n|\Lambda(\tilde{X}^n))}\right]\nonumber \\
	&\leq \log\frac{1}{\beta}+\mbb E\left[\log\frac{P_{X^n,K,Y^n}(\tilde{X}^n,\tilde{K},\tilde{Y}^n)}{P_{\tilde{X}^n,\tilde{K},\tilde{Y}^n}(\tilde{X}^n,\tilde{K},\tilde{Y}^n)}\right] \nonumber \\
	&=\log\frac{1}{\beta}-D(P_{\tilde{X}^n,\tilde{K},\tilde{Y}^n}||P_{X^n,K,Y^n})\nonumber \\
	&\leq \log\frac{1}{\beta}.
\nonumber \end{align}
Therefore, it follows that

\begin{align} 
&\frac{1}{n} I(\tilde{K};\tilde{Z}^{n}|\tilde{Y}^{n})\nonumber \\
&\leq 
\frac{1}{n}\mbb E\left[\log\frac{P_{T^nZ^n}(\Lambda(\tilde{X}^n),\tilde{Z}^n)}{P_{T^n}(\Lambda(\tilde{X}^n))P_{Z^n}(\tilde{Z}^n)}\right]+\frac{1}{n}\log\frac{1}{\beta} \nonumber \\
&\overset{(a)}{\leq} R_{\epsilon,\sup}+\mu+\frac{1}{n}\log\frac{1}{\beta},
\nonumber
\end{align}
where $(a)$ follows from \eqref{upperboundnewmutinf}.
\end{proof}
\begin{claimproof}
We have
\begin{equation}
\frac{1}{n}H(\tilde{K}|\tilde{Y}^{n})=\frac{1}{n}I(\tilde{K};\tilde{Z}^{n}|\tilde{Y}^{n})+\frac{1}{n}H(\tilde{K}|\tilde{Y}^{n},\tilde{Z}^{n}). \nonumber
\end{equation}

Now, since $\tilde{K}$ is equal to $\tilde{L}=\Psi(\tilde{Y}^n,\tilde{Z}^n)$ with probability one, it holds that
$H(\tilde{K}|\tilde{Y}^n,\tilde{Z}^n)=0.$
It follows using Lemma \ref{upperboundlogcardinality} and Lemma \ref{mutinftilektildezcondtildey}  that $\frac{1}{n}H(K|Y^n)\leq R_{\epsilon,\sup}+\zeta(n,\epsilon,\beta,\mu)$  for any $\beta \in \mathcal{B}$  and for sufficiently large $n.$
Since $R_{\epsilon,\sup}\leq u(\epsilon),$ it follows that $$\frac{1}{n}H(K|Y^n)\leq u(\epsilon)+\zeta(n,\epsilon,\beta,\mu).$$
 This completes the proof of the claim.
\end{claimproof}
 
 Now, let $J$ be a random variable uniformly distributed on $\{1,\dots, n\}$ and independent of $K$, $X^n$ and $Y^n$. We further define $U=(K,X_{1},\dots, X_{J-1},Y_{J+1},\dots, Y_{n},J).$ It holds that $U \circlearrow{X_J} \circlearrow{Y_J}.$
 Notice  that
 {{\begin{align}
 		H(K)&\overset{(a)}{=}H(K)-H(K|X^{n})\nonumber\\
 		&=I(K;X^{n}) \nonumber\\
 		&\overset{(b)}{=}\sum_{i=1}^{n} I(K;X_{i}|X_{1},\dots, X_{i-1}) \nonumber\\
 		&=n I(K;X_{J}|X_{1},\dots, X_{J-1},J) \nonumber\\
 		&\overset{(c)}{\leq }n I(U;X_{J}), \label{upperboundentropyK}
 		\end{align}}}where $(a)$ follows because $K=\Phi(X^n)$ and $(b)$ and $(c)$ follow from the chain rule for mutual information.
 		
 Let us now introduce the following lemma:
  \begin{lemma} (Lemma 17.12 in \cite{codingtheorems})
 	For arbitrary random variables $S$ and $R$ and sequences of random variables $X^{n}$ and $Y^{n}$, it holds that
 	\begin{align}
 	&I(S;X^{n}|R)-I(S;Y^{n}|R) \nonumber \\
 	&=\sum_{i=1}^{n} I(S;X_{i}|X_{1},\dots, X_{i-1}, Y_{i+1},\dots, Y_{n},R) \nonumber \\ &\quad -\sum_{i=1}^{n} I(S;Y_{i}|X_{1},\dots, X_{i-1}, Y_{i+1},\dots, Y_{n},R) \nonumber \\
 	&=n[I(S;X_{J}|V)-I(S;Y_{J}|V)],\nonumber
 	\end{align}
 	where $V=(X_{1},\dots, X_{J-1},Y_{J+1},\dots, Y_{n},R,J)$, with $J$ being a random variable independent of $R$,\ $S$, \ $X^{n}$ \ and $Y^{n}$ and uniformly distributed on $\{1 ,\dots, n \}$.
 	\label{lemma1}
 \end{lemma}
 Applying Lemma \ref{lemma1} for $S=K$, $R=\varnothing$ with $V=(X_1,\hdots, X_{J-1},Y_{J+1},\hdots, Y_{n},J)$ yields
 {{\begin{align}
 		&I(K;X^{n})-I(K;Y^{n}) \nonumber \\
 		&=n[I(K;X_{J}|V)-I(K;Y_{J}|V)] \nonumber\\
 		&\overset{(a)}{=}n[I(KV;X_{J})-I(V;X_{J})-I(KV;Y_{J})+I(V;Y_{J})] \nonumber\\ 
 		&\overset{(b)}{=}n[I(U;X_{J})-I(U;Y_{J})], 
 		\label{UhilfsvariableMIMO1}
 		\end{align}}}where $(a)$ follows from the chain rule for mutual information and from the fact that $V$ is independent of $(X_{J},Y_{J})$ and $(b)$ follows from $U=(K,V)$. It results using (\ref{UhilfsvariableMIMO1}) that
 {\begin{align}
 		n[I(U;X_{J})-I(U;Y_{J})]
 		&=I(K;X^{n})-I(K;Y^{n}) \nonumber\\
 		&=H(K)-I(K;Y^{n})\nonumber \\ 
 		&=H(K|Y^n).
 		\nonumber
 		\end{align}}
It follows using  Claim \ref{claim1} that for infinitely large $n$ satisfying \eqref{choiceepsilon}
\begin{align}
    I(U;X_{J})-I(U;Y_{J}) \leq u(\epsilon)+\zeta(n,\epsilon,\beta,\mu). \label{upperboundmutinf}
\end{align}
Since the joint distribution of $X_{J}$ and $Y_{J}$ is equal to $P_{XY},$ it follows from \eqref{upperboundentropyK} using \eqref{upperboundmutinf} that $\frac{H(K)}{n}$ is upper-bounded by $I(U;X)$ subject to $I(U;X)-I(U;Y) \leq u(\epsilon) + \zeta(n,\epsilon,\beta,\mu)$ with $U$ satisfying $U \circlearrow{X} \circlearrow{Y}.$ As a result, for sufficiently large $n$ satisfying \eqref{choiceepsilon}, it follows using \eqref{ratecorrelated} that any achievable $\epsilon$-UCR rate satisfies
\begin{align}
H <\underset{ \substack{U \\{\substack{U \circlearrow{X} \circlearrow{Y}\\ I(U;X)-I(U;Y) \leq u(\epsilon)+\zeta(n,\epsilon,\beta,\mu)}}}}{\max} I(U;X)+\delta.
 \label{righthandsideconverse}
\end{align}
By taking the limit when $n$ tends to infinity and then the infinimum over all $\beta\in\mathcal{B},\mu\in \mathcal{E},\delta>0,$ of the right-hand side of \eqref{righthandsideconverse}, it follows that
\begin{align}
&H \leq \underset{ \substack{U \\{\substack{U \circlearrow{X} \circlearrow{Y}\\ I(U;X)-I(U;Y) \leq u(\epsilon)}}}}{\max} I(U;X). \nonumber
\end{align}
This completes the proof of the upper-bound on the $\epsilon$-UCR capacity.

\section{Conclusion}
\label{conclusion}
In our work, we established a general expression for a lower and upper bound on the $\epsilon$-UCR capacity for a standard two-source model with unidirectional communication over arbitrary point-to-point channels. The bounds hold with equality except possibly at the points of discontinuity of the $\epsilon$-transmission capacity of the single-user channel, of which there are, at most, countably many.
As a future work, it would be interesting to investigate the problem of UCR generation from i.i.d. finite sources with two-way communication over arbitrary point-to-point channels.
\section*{Acknowledgments}
\color{black}
The authors acknowledge the financial support by the Federal Ministry of Education and Research
of Germany (BMBF) in the programme of “Souverän. Digital. Vernetzt.”. Joint project 6G-life, project identification number: 16KISK002.
H. Boche and R. Ezzine were further supported in part by the BMBF within the national initiative on Post Shannon Communication (NewCom) under Grant 16KIS1003K. 
C.\ Deppe was further supported in part by the BMBF within the national initiative on Post Shannon Communication (NewCom) under Grant 16KIS1005. C. Deppe was also supported by the German Research Foundation (DFG) within the project DE1915/2-1. M. Wiese was supported by the Bavarian Ministry of
Economic Affairs, Regional Development and Energy as part of the project 6G
Future Lab Bavaria.
\color{black}



\appendix
\section{Auxiliary Lemmas}
\label{auxlemmas}
\begin{lemma}
\label{upperboundvariance}
For $\lvert \mc K \rvert \geq 3,$ it holds for sufficiently large $n$ that
 \begin{align}
\mathrm{var}\left[\frac{1}{n}\log\frac{1}{P_{K}(K)}\right]\leq \lambda(\beta).
\nonumber \end{align} 
\end{lemma}
\begin{proof}
We have
\begin{align}
    \mbb E \left[\log^{2} P_{K}(K)\right]
    =\frac{1}{\ln(2)^2}\mbb E \left[\ln^{2} P_{K}(K)\right].
\nonumber \end{align}

Define the set
$\mc K_{L}=\{k\in \mc K: P_{K}(k)\leq \frac{1}{e}  \}$ and the set
$\mc K_{H}=\{k\in \mc K: P_{K}(k)> \frac{1}{e}  \}.$
Clearly, it holds that $\lvert \mc K_{L} \rvert +\lvert \mc K_{H} \rvert =\lvert \mc K\rvert.$ 
Let 
$
    P_{L}=\sum_{k\in \mc K_{L}} P_{K}(k)$
and
$P_{H}=\sum_{k\in \mc K_{H}} P_{K}(k).$

Notice first that $\lvert \mc K_{H} \rvert \frac{1}{e}<P_{H}\leq 1.$
This yields $\lvert \mc K_{H} \rvert<e.$
Therefore, it holds that $\lvert \mc K_{H} \rvert\leq 2.$
Since $\lvert \mc K\rvert \geq 3,$ it follows that
$\lvert \mc K _L \rvert=\lvert \mc K \rvert - \lvert \mc K_H\rvert \geq 1.$

Now, it holds that
\begin{align}
    &\mbb E \left[\ln^{2} P_{K}(K)\right]\nonumber \\
    &=\sum_{k\in \mc K_{L}} P_{K}(k)\ln^2\frac{1}{P_{K}(k)}+ \sum_{k\in \mc K_{H}} P_{K}(k)\ln^2\frac{1}{P_{K}(k)}.
\label{termsinthesum} 
\end{align}

We will find appropriate upper-bound for each term in the right-hand side of \eqref{termsinthesum}. On the one hand, we have
\begin{align}
    &\sum_{k\in \mc K_{L}} P_{K}(k)\ln^2\left(\frac{1}{P_{K}(k)}\right) \nonumber \\
    &\overset{(a)}{\leq} P_{L}\ln^2\left( \sum_{k\in \mc K_{L}} \frac{P_{K}(k)}{P_{L}}\frac{1}{P_{K}(k)}\right) \nonumber \\
    &=P_{L}\ln^{2}\frac{\lvert \mc K_{L}\rvert}{P_{L}},\nonumber\end{align}
where $(a)$ follows because $\ln^{2}(y)$ is concave in the range $y\geq e$ and because for any $k \in \mc K_L,$  $\frac{1}{P_{K}(k)}\geq e.$

On the other hand, we have
\begin{align}
    &\sum_{k\in \mc K_{H}} P_{K}(k)\ln^2\frac{1}{P_{K}(k)} \nonumber\\
    &\overset{(a)}{\leq} \sum_{k\in \mc K_{H}}  P_{K}(k) \ln^2(e) \nonumber \\
    &\leq 1,
\nonumber \end{align}
where $(a)$ follows because $\ln^2(1/y)$ is non-increasing in the range $0<y\leq 1$ and because $\frac{1}{e}<P_{K}(k)\leq 1$ for $k\in \mc K_{H}.$

This implies using the fact that $\lvert \mc K \rvert \geq \lvert \mc K_{L}\rvert\geq 1$
that
\begin{align}
    &\mbb E \left[\ln^{2} P_{K}(K)\right] \nonumber \\
    &\leq 1+ P_{L}\ln^{2}\frac{\lvert \mc K_L\rvert}{P_{L}} \nonumber \\
    &=1+P_{L}\left(\ln\left(\lvert\mc K_L\rvert \right)+\ln\frac{1}{P_{L}}  \right)^{2} \nonumber \\
    &\leq 1+P_{L}\left(\ln\left(\lvert\mc K\rvert\right)+\ln\frac{1}{P_{L}}  \right)^{2} \nonumber \\
    &= 1+P_{L}\left(\ln\left(\lvert\mc K\rvert\right)^2+\ln^{2}\frac{1}{P_{L}}+2\ln\left(\frac{1}{P_{L}}\right) \ln\lvert\mc K\rvert \right) \nonumber \\
    &\overset{(a)}{\leq} 1+\ln\left(\lvert\mc K\rvert\right)^2+\frac{4}{e^2}+2\frac{1}{e}\ln\lvert\mc K\rvert,
\nonumber \end{align}
where $(a)$ follows because $y\ln^2(1/y)$ and $y\ln(1/y)$ are maximized by $\frac{4}{e^2}$ and $\frac{1}{e}$ in the range $0<y\leq 1,$ respectively.

Thus, it follows that 
\begin{align}
    &\mbb E \left[\frac{1}{n^2}\log^{2} P_{K}(K)\right] \nonumber \\
    &\leq\frac{1}{n^2\ln(2)^2}\left(1+\ln\left(\lvert\mc K\rvert\right)^2+\frac{4}{e^2}+2\frac{1}{e}\ln\lvert\mc K\rvert\right) \nonumber \\
    &\overset{(a)}{\leq} \frac{1+\frac{4}{e^2}}{n^2\ln(2)^2}+\frac{\log^{2}\left(\lvert\mc K\rvert\right)}{n^2}+\frac{2c}{n\ln(2)e}, \nonumber 
\nonumber \end{align}
where $(a)$ follows because $\frac{\log \lvert \mc K \rvert}{n} \leq c$ 
(from \eqref{cardinalitycorrelated}). 

Since $\underset{n\rightarrow\infty}{\lim}\frac{1+\frac{4}{e^2}}{n^2\ln(2)^2}+\frac{2c}{n\ln(2)e}=0,$ it follows that for sufficiently large $n$
\begin{align}
      &\mbb E \left[\frac{1}{n^2}\log^{2} P_{K}(K)\right] \leq \beta+\frac{\log^{2}\left(\lvert\mc K\rvert\right)}{n^2}. \nonumber
\end{align}

From \eqref{uniformity}, we know that
\begin{align}
  \frac{\log\lvert \mc K \rvert}{n}\leq \frac{H(K)}{n}+\beta.
\nonumber \end{align}
It follows that
\begin{align}
     \mbb E \left[\frac{1}{n^2}\log^{2} P_{K}(K)\right] \leq \beta+\frac{1}{n^2} ( H(K)+n\beta)^2
\nonumber \end{align}
which yields
\begin{align}
    &\mathrm{var}\left[\frac{1}{n}\log\frac{1}{P_{K}(K)}\right] \nonumber \\
    &=\mbb E\left[\frac{1}{n^2}\log^2\left(\frac{1}{P_{K}(K)}\right)\right]-\frac{1}{n^2} H(K)^2 \nonumber \\
    &\leq \beta+ 2\beta \frac{H(K)}{n}+\beta^2 \nonumber \\
    &\overset{(a)}{\leq} \beta+ 2\frac{\beta\log\lvert \mc K \rvert}{n}+\beta^2\nonumber \\
    &\overset{(b)}{\leq} \beta+ 2\beta c+\beta^2 \nonumber \\
    &=\lambda(\beta),
\nonumber \end{align}
where $(a)$ follows because $H(K)\leq \log\lvert \mc K \rvert$ and $(b)$ follows from \eqref{cardinalitycorrelated}.
\end{proof}

\begin{lemma}
\label{boundprobsett} Let
\begin{align}
    \mc L=\Big\{ k\in \mc K: \frac{1}{n}\log\frac{1}{P_{K}(k)} \geq \frac{1}{n}H(K) -\frac{\gamma(\epsilon,\beta)}{2}     \Big\}.
\nonumber \end{align}For sufficiently large $n$ and for $0<\lambda(\beta)<1,$ we have
\begin{align}
\mbb P\left[K\in \mc L \right]\geq 1-4\frac{\lambda(\beta)}{\gamma(\epsilon,\beta)^2}>0. \nonumber
\end{align}
\end{lemma}
\begin{proof}
It holds that
\begin{align}
    &\mbb P\left[ K\notin \mc L\right] \nonumber \\
    &=\mbb P\left[ \frac{1}{n}\log\frac{1}{P_{K}(K)} - \frac{1}{n}H(K) <-\frac{\gamma(\epsilon,\beta)}{2}  \right] \nonumber \\
    &\leq \mbb P\left[ \bigg\lvert \frac{1}{n}\log\frac{1}{P_{K}(K)} - \frac{1}{n}H(K) \bigg\rvert > \frac{\gamma(\epsilon,\beta)}{2}  \right] \nonumber \\
    &\overset{(a)}{\leq} 4\frac{\mathrm{var}\left[\frac{1}{n}\log\left(\frac{1}{P_{ K}\left(K\right)}\right)  \right]}{\gamma(\epsilon,\beta)^2} \nonumber \\
    &\overset{(b)}{\leq} 4\frac{\lambda(\beta)}{\gamma(\epsilon,\beta)^2}, \nonumber
\end{align}
where $(a)$ follows from Chebyshev's inequality since $\mbb E\left[\frac{1}{n}\log\frac{1}{P_{K}(k)}\right] =\frac{1}{n}H(K)$ and $(b)$ follows from Lemma \ref{upperboundvariance}.
Therefore, we have
\begin{align}
\mbb P\left[K\in \mc L \right]&\geq 1-4\frac{\lambda(\beta)}{\gamma(\epsilon,\beta)^2}\overset{(a)}{>}0, \nonumber
\end{align}
where $(a)$ follows because for $0<\lambda(\beta)<1,$ we have \begin{align}
	0<4\frac{\lambda(\beta)}{\gamma(\epsilon,\beta)^2}= \sqrt{\lambda(\beta)}(1-\sqrt{\epsilon}) <1-\sqrt{\epsilon}<1. \nonumber
	\end{align}

This proves Lemma \ref{boundprobsett}.
\end{proof}

\begin{lemma}
\label{boundprobsetd}
For sufficiently large $n$ and for $0<\lambda(\beta)<1,$ it holds that
\begin{align}
  \mbb P \left[(K,Y^n)\in \mc D\right]  \geq \left(1-4\frac{\lambda(\beta)}{\gamma(\epsilon,\beta)^2}\right)^2. \nonumber
\end{align}
\end{lemma}
\begin{proof}
Let
\begin{align}
    \mc L=\Big\{ k\in \mc K: \frac{1}{n}\log\frac{1}{P_{K}(k)} \geq \frac{1}{n}H(K) -\frac{\gamma(\epsilon,\beta)}{2}     \Big\}.
\nonumber \end{align}
We have
\begin{align}
    &\mbb P \left[(K,Y^n)\in \mc D\right]\nonumber \\ &\geq \sum_{k \in \mc L} \mbb P \left [(K,Y^n)\in \mc D|K=k\right]P_{K}(k) \nonumber \\
    &=\sum_{k \in \mc L}  \mbb P\left[\frac{\log\frac{1}{P_{K|Y^n}(k|Y^n)} }{n}\geq \frac{H(K|Y^n)}{n}-\gamma(\epsilon,\beta)  \right] P_{K}(k) \nonumber \\
    &=\sum_{k \in \mc L}  \mbb P\left[P_{K|Y^n}(k|Y^n)\leq 2^{n\gamma(\epsilon,\beta)-H(K|Y^n)} \right] P_{K}(k) \nonumber \\
    &\overset{(a)}{\geq}\sum_{k\in \mc L}\left(1-\frac{P_{K}(k)}{2^{n\gamma(\epsilon,\beta)-H(K|Y^n)}}\right) P_{K}(k) \nonumber \\
    &\overset{(b)}{\geq}\left(1-2^{\left[-n\frac{\gamma(\epsilon,\beta)}{2}+H(K|Y^n)-H(K)\right]}\right) \mbb P\left[ K\in \mc L \right] \nonumber \\
    &\overset{(c)}{\geq}\left(1-2^{-n\frac{\gamma(\epsilon,\beta)}{2}}\right)\mbb P\left[ K\in \mc L \right]  \nonumber \\
    &\overset{(d)}{\geq}\left(1-2^{-n\frac{\gamma(\epsilon,\beta)}{2}}\right)\left(1-\frac{4\lambda(\beta)}{\gamma(\epsilon,\beta)^2}\right), \nonumber
   \end{align}
  where $(a)$ follows from Markov's inequality since $P_{K}(k)=\mbb E\left[P_{K|Y^n}(k|Y^n)\right],$ $(b)$ follows because for $k\in \mc L,$  we know that $ P_{K}(k) \leq 2^{n\frac{\gamma(\epsilon,\beta)}{2}-H(K)},$ $(c)$ follows because $H(K|Y^n)-H(K)\leq 0$ and $(d)$ follows from Lemma \ref{boundprobsett}.
  
  Since $\underset{n\rightarrow \infty}{\lim} 1-2^{-n\frac{\gamma(\epsilon,\beta)}{2}}=1,$ it follows that for sufficiently large $n$
  \begin{align}
  \mbb P \left[(K,Y^n)\in \mc D\right]  \geq \left(1-4\frac{\lambda(\beta)}{\gamma(\epsilon,\beta)^2}\right)^2. \nonumber
  \end{align}
\end{proof}
\end{document}